\documentclass[12pt]{article}

\usepackage{amsmath,amssymb,amsthm}

\numberwithin{equation}{section}

\newtheorem{theorem}{Theorem}[section]

\newtheorem{lemma}{Lemma}[section]

\theoremstyle{definition}

\allowdisplaybreaks

\begin{document}

\title{Adiabatics using quantum action}

\author{M.~V.~Karasev
\\
{\small Department of Applied Mathematics}
\\
{\small National Research University Higher School of Economics}
\\
{\small E-mail: mkarasev@hse.ru}}

\date{}

\maketitle

\begin{abstract}
For slow--fast quantum systems, we compute first corrections
to the quantum action and to the effective slow Hamiltonian.
\end{abstract}

\section{Introduction}

The phase space of a classical adiabatic system is represented as
the direct product of a ``fast'' fiber and a ``slow'' base
which are distinguished from each other algebraically
by the presence of a small parameter
(let us call it $\varepsilon$)
in the mutual Poisson brackets between coordinates on the slow phase manifold.
Quantum adiabatic systems are organized in a similar way
but the notion of Poisson brackets is replaced, in this case,
by the notion of commutator, thus the mutual commutators
between slow quantum coordinates
in such systems are proportional to the small parameter~$\varepsilon$.

The guiding idea of the adiabatic approximation method
can be declared as follows:
try, asymptotically in $\varepsilon$, to extract quantum mechanics
along slow coordinates only
and clean it from fast coordinates as much as possible.

Since the fast quantum eigenstates depend on slow coordinates,
then there generally appears a correlation between fast and slow ``directions'' already
at the first order of the adiabatic asymptotic expansion
with respect to~$\varepsilon$.
The corrections of order $\varepsilon$
to the leading Born--Oppenheimer adiabatic terms can be computed
in several different manners~[1--6].

The problem is that these corrections are operators still acting in fast directions,
e.g., they interlace different adiabatic energy levels, and so,
one cannot say that we have up to $O(\varepsilon^2)$ an effective quantum adiabatic Hamiltonian
over the slow coordinate algebra.

The known schemes (see, e.g., [7--11]) allowing to make these
$\varepsilon$-corrections to be operators in slow directions work well
if the fast direction is low-dimensional and described by matrices
(as it happens for the Pauli, Dirac, or Maxwell equations).
But if the fast direction is infinite-dimensional
and described by differential operators
(as it happens for systems with light and heavy particles,
or for systems with fast Larmor vortices, etc.), then
the matrix point of view fails.
It is very ineffective, nongeometric, and nonalgebraic to deal with
infinite-dimensional matrices; the situation here is very similar to that
of old  matrix quantum mechanics.

Another inconvenience of the usual adiabatic approaches
is that they are based on the information
about the fast eigenstates of the system.
If these states are just eigenvectors of 2-, 3-,\dots, 6-dimensional matrices,
then this is not a bid difficulty.
But if one needs to compute eigenstates of a differential operator,
and variations of these fast states provide contributions
to the final effective Hamiltonian,
then such an adiabatic approach occurs to be unsatisfactory.

At the same time in classical mechanics, there are adiabatic approximation schemes
[12--17] which allow, under some conditions, to exclude fast directions and compute
the effective adiabatic Hamiltonian on the slow phase manifold
in an elegant geometric manner.

The aim in the present paper is to develop the quantum version of these classical adiabatic
approximation schemes. We continue and generalize the approach of [18]
introducing the notion of quantum action.
The result can, in short, be formulated as follows:
the adiabatic Hamiltonians act on the eigensubspaces of the quantum
action.
Thus, in our work, the universal Plank concept of ``action'' replaces
the customary method of matrix diagonalization
and provides an opportunity for solving the quantum adiabatic problem
in the case of infinite-dimensional fast direction.

The main objects taken from [18] are noncommutative-deformed slow coordinates
which commute with the quantum action and satisfy the canonical commutation relations.
They can be considered
as a certain analog of the Landau guiding center coordinates.
Our final effective adiabatic Hamiltonian
is expressed via these noncommutative-deformed slow coordinates.
So, we do not keep the initial slow phase manifold unchanged,
but deform it.
This deformation is a necessary step to separate the fast directions
and to extract the slow dynamics.

The averaged values of the noncommutative-deformed slow coordinates, in general,
do not coincide with the classical ones.
Thus the above-mentioned analogy with the Landau guiding center coordinates is not complete.
In our case, the deformation indeed changes the average geometry of the slow space.
One has to stress that the word ``geometry'' is not used occasionally here:
the noncommutative-deformed slow coordinates
are changed by the usual differential geometric rule
under changing the  classical slow coordinates.

Note that all these deformations and computations of asymptotics
we make up to $O(\varepsilon^2)$ only for simplicity.
Also note that the presence of the small parameter $\varepsilon$
in the commutators between slow coordinates implies that one can apply
the semiclassical asymptotics theory to analyze the obtained effective quantum
adiabatic Hamiltonian.

Moreover, one can try to develop the whole quantum adiabatic
approximation scheme from the very beginning
in combination with the semiclassical asymptotics (as in [18]).
But in the present paper we do not
follow this way; the advantages of our general approach is that we
do not restrict ourselves to specific assumptions required to
apply the semiclassical theory (like the integrability,
compactness or noncompactness of energy levels, absence of
separatrix, etc.).

Another important remark:
we do not assume that in fast coordinates the system has any semiclassical behavior.
So our general results do not use any Poisson or symplectic geometry
along fast fibers.
But, of course, we can also consider the totally semiclassical case:
with an additional parameter $\hbar$ at fast commutators,
and the corresponding small parameter $\varepsilon\hbar$
at slow commutators.
In this case, our general quantum results make possible
the more systematic derivation and reorganization
of the Plank-type quantization conditions previously obtained in [18]
in the framework of the semiclassical adiabatic approximation.

In the classical limit $\hbar\to0$, when the commutator
$\frac{i}{\hbar}[\cdot,\cdot]$ becomes the Poisson brackets,
our quantum adiabatic approximation results coincide with those
known in classical mechanics.

\section{Adiabatic approximation via quantum ``action''}

Let us consider the case of one degree of freedom in the fast direction.

We denote by $\mathcal{A}=\mathcal{A}(\mathbf{R}^2)$
an algebra of ``fast'' operators of the type $\mathbb{F}=F(\mathbb{A})$,
where the function $F$ with all derivatives
has no more than a polynomial growth at infinity, and the set
$\mathbb{A}=(\mathbb{A}_1,\mathbb{A}_2)$ consists of two self-adjoint operators
$\mathbb{A}_1$, $\mathbb{A}_2$
representing the Heisenberg algebra in some Hilbert space, i.e.,
\begin{equation}
[\mathbb{A}_1,\mathbb{A}_2]=i.
\end{equation}
These operators are assumed to be Weyl-symmetrized when substituted into the function $F$,
that is,
\begin{equation}
F(\mathbb{A})\overset{\text{def}}{=}
\frac1{2\pi}\int_{\mathbf{R}^2}\tilde{F}(t)\exp\{it\cdot\mathbb{A}\}\,dt,
\qquad
\tilde{F}(t)\overset{\text{def}}{=}\int_{\mathbf{R}^2}F(A) e^{-itA}\,dA.
\end{equation}
The function $F$ is called the Weyl symbol of the operator $F(\mathbb{A})$.

We also consider a ``slow'' phase space $\mathcal{D}\approx \mathbf{R}^{2n}$
with the set of Darboux coordinates $x=(x^1,\dots,x^{2n})$, so that
\begin{equation}
\{x^j,x^l\}=J^{jl},\qquad
J=\left(\begin{matrix}0&-I\\I&0\end{matrix}\right)
\end{equation}
The ``slow'' quantum operators $\hat{x}^j$ are assumed to obey
the commutation relations with a small parameter $\varepsilon$:
\begin{equation}
[\hat{x}^j,\hat{x}^l]=-i\varepsilon J^{jl},
\end{equation}
and these operators commute with all fast operators
\begin{equation}
[\hat{x}^j,\mathbb{A}_k]=0.
\end{equation}

The quantum adiabatic system which we deal with has the Hamiltonian
\begin{equation}
\hat{\mathbb{H}}=\mathbb{H}(\hat{x}),\qquad
\mathbb{H}(x)\in\mathcal{A},
\end{equation}
where the slow operators $\hat{x}=(\hat{x}^1,\dots,\hat{x}^{2n})$
are Weyl-symmetrized.

The function $\mathbb{H}=\mathbb{H}(x)$ in (2.6) is usually called
the ``operator-valued symbol''.
We are here interested in the case where $\mathbb{H}(x)$, at each $x$,
has discrete spectrum and the eigenvalues $\lambda_k(x)$ called ``terms''
are smooth functions of $x$.

Note that all terms $\lambda_k(x)$ are automatically nondegenerate, i.e.,
their multiplicity equals~$1$, since the algebra (2.1) has only one degree of freedom.

The adiabatic approximation is, generally speaking, targeted to a reduction
of the Hamiltonian (2.6) over the direct product of fast and slow algebras (2.1), (2.4)
to some effective Hamiltonian over the slow algebra only.
The first step is the reduction to the quantized terms:
\begin{equation}
\hat{\lambda}_k+O(\varepsilon),\qquad
\hat{\lambda}_k\overset{\text{def}}{=}\lambda_k(\hat{x}),
\end{equation}
and secondly to compute the higher order corrections to the leading terms
$\lambda_k$ (the effective adiabatic Hamiltonians):
\begin{equation}
\hat{\lambda}_k+\varepsilon\hat{\mu}_k+O(\varepsilon^2).
\end{equation}

The rough reduction from (2.6) to (2.7) is made, for instance,
by the Born--Oppenheimer method.
But more precise reduction to (2.8) and explicit computation
of the scalar Hamiltonians $\mu_k$ over the slow phase space is a not simple question.

Let us ask: How one can obtain the terms $\lambda_k(x)$
from the operator-valued symbol $\mathbb{H}(x)$?
The routine answer is to use the eigenprojection
$\mathbb{P}_k(x)$ (or the eigenvectors, as in the original
Born--Oppenheimer scheme):
\begin{equation}
\mathbb{H}(x)\mathbb{P}_k(x)=\lambda_k(x)\mathbb{P}_k(x),
\qquad k=0,1,2,\dots\,.
\end{equation}
By quantizing (2.9) and using relations (2.4) one derives
\begin{equation*}
\hat{\mathbb{H}}\hat{\mathbb{P}}_k
=\hat{\mathbb{P}}_k\hat{\lambda}_k+O(\varepsilon).
\end{equation*}
Thus, by the intertwining operator $\hat{\mathbb{P}}_k=\mathbb{P}_k(\hat{x})$,
the Hamiltonian (2.6) is transformed to (2.7).
But the next step, in the order $\varepsilon$,
fails to reach (2.8) on this way.

The defect is that each relation (2.9) deals only with a separate labelling number $k$
and does not allow one to join and consider all these numbers simultaneously.

A natural object combining all the labelling numbers $k$ from (2.9) is the family
of action operators
\begin{equation*}
\mathbb{S}(x)\overset{\text{def}}{=}\sum_k k\cdot\mathbb{P}_k(x),\qquad x\in\mathcal{D}.
\end{equation*}
It has the integer-valued spectrum and its one-parameter group satisfies
the $2\pi$-periodicity condition:
\begin{equation}
\exp\{2\pi i\mathbb{S}(x)\}=\mathbf{1}.
\end{equation}
Here $\mathbf{1}$ is the unity element (operator) in $\mathcal{A}$.

Let us assume that the operator-valued symbol $\mathbb{H}(x)$ is bounded
from below and its eigenvalues are ordered as follows:
\begin{equation*}
\lambda_0(x)<\lambda_1(x)<\dots\,.
\end{equation*}
In the same way, one can consider the case where the eigenvalues
are not bounded from below.

We choose a smooth function $f=f(s,x)$, strongly monotonic in $s$, such that
\begin{equation}
\lambda_k(x)=f(k,x).
\end{equation}
Then one obtains the representation
\begin{equation}
\mathbb{H}(x)=f(\mathbb{S}(x),x).
\end{equation}

Actually such a representation of the fast Hamiltonian one can have a priori,
i.e., without any references to the eigenprojections,
just keeping the key property (2.10) of the action operators.
Such an action determined by (2.10) is
the operator analog of labelling numbers.
And representation (2.12) can in fact be considered as a fundamental property
of the system. Representations of this type are widely used
in semiclassical and classical mechanics.

Now let us quantize relations (2.10), (2.12) by substituting  $x\to\hat{x}$.
We assume everywhere below, without special mentioning,
that all used scalar or $\mathcal{A}$-valued functions in slow coordinates
are smooth and belong to suitable classes allowing one to quantize these functions
over the Heisenberg algebra (2.4).

Note that, after quantization, the operator
$\hat{\mathbb{S}}=\mathbb{S}(\hat{x})$
does not obey relation (2.10) exactly,
the discrepancy order is $O(\varepsilon)$.
But one can try to construct a correct ``quantum action''
\begin{equation}
\hat{\mathbb{S}}_\varepsilon=\mathbb{S}(\hat{x})+\varepsilon\mathbb{W}(\hat{x})
+O(\varepsilon^2)
\end{equation}
to preserve the periodicity:
\begin{equation}
\exp\{2\pi i\hat{\mathbb{S}}_\varepsilon \}=\mathbf{1}.
\end{equation}
Then the eigenvalues of $\hat{\mathbb{S}}_\varepsilon$
are integer: $\operatorname{spectrum}(\hat{\mathbb{S}}_\varepsilon)
=\{k\mid k=0,1,2,\dots\}$
and the eigensubspaces of $\hat{\mathbb{S}}_\varepsilon$
can be used to separate adiabatic bands of the
Hamiltonian $\hat{\mathbb{H}}$.

How $\hat{\mathbb{H}}$ is related to $\hat{\mathbb{S}}_\varepsilon$?
From (2.12) and (2.13) we have
$$
\mathbb{H}(x)=f(\mathbb{S}_\varepsilon(x),x)
-\varepsilon\partial f(\mathbb{S}_\varepsilon,x)\mathbb{W}(x)+O(\varepsilon^2),
$$
where $\partial\overset{\text{def}}{=}\partial/\partial s$
and the Weyl symmetrization of $\mathbb{S}_\varepsilon$ and $\mathbb{W}$ is assumed.
Therefore, after quantization,
\begin{equation}
\hat{\mathbb{H}}=f(\hat{\mathbb{S}}_\varepsilon,\hat{x})
+\varepsilon\hat{\mathbb{Q}}
-\varepsilon\partial f(\hat{\mathbb{S}}_\varepsilon,\hat{x})
\hat{\mathbb{W}}+O(\varepsilon^2),
\end{equation}
where on the right-hand side of (2.15) all the operators are Weyl-symmetrized,
and the correction $\mathbb{Q}$ is described in Appendix~A.

In (2.15) we know, up to $O(\varepsilon^2)$, the spectrum of
$\hat{\mathbb{S}}_\varepsilon$, but the difficulty is
that $\hat{x}$ and $\hat{\mathbb{W}}$ do not commute
with $\hat{\mathbb{S}}_\varepsilon$ and so it is impossible to restrict
the right-hand side of (2.15) to eigensubspaces
of $\hat{\mathbb{S}}_\varepsilon$.

To eliminate this obstruction in the order $\varepsilon$,
let us first deform the slow coordinates
\begin{equation}
\mathbb{X}_\varepsilon=x+\varepsilon \mathbb{Y}+O(\varepsilon^2)
\end{equation}
under the condition
\begin{equation}
[\hat{\mathbb{S}}_\varepsilon,\hat{\mathbb{X}}_\varepsilon]=O(\varepsilon^2).
\end{equation}
Here $\mathbb{Y}$ is some vector field on $\mathcal{D}$
with values in $\mathcal{A}$.

Then (2.15) is transformed to
\begin{equation}
\hat{\mathbb{H}}=f(\hat{\mathbb{S}}_\varepsilon,\hat{\mathbb{X}}_\varepsilon)
+\varepsilon\Big(\hat{Q}
-\partial f(\hat{\mathbb{S}}_\varepsilon,\hat{\mathbb{X}}_\varepsilon)\hat{\mathbb{W}}
-Df(\hat{\mathbb{S}}_\varepsilon,\hat{\mathbb{X}}_\varepsilon)\hat{\mathbb{Y}}\Big)
+O(\varepsilon^2),
\end{equation}
where $D=\partial/\partial x$ and all operators in (2.18) are Weyl-symmetrized.

The leading term in (2.18) can already be restricted (up to $O(\varepsilon^2)$)
to the $k$th eigensubspace of $\hat{\mathbb{S}}_\varepsilon$
and provides the reduction of the Hamiltonian (2.6)
to $f(k,\hat{\mathbb{X}}_\varepsilon)+O(\varepsilon)$
with slow coordinates $x$ replaced by
their quantum deformations~$\mathbb{X}_\varepsilon$.

Now the following two questions arise:

--- What are mutual commutation relations between $\mathbb{X}_\varepsilon^j$?
Can these deformed coordinates be actually used
as analogs of slow classical coordinates?

--- How can one deal with the $\varepsilon$-order correction in (2.18) which still
does not commute with $\hat{\mathbb{S}}_\varepsilon$?

\section{Noncommutative-deformed slow coordinates}

From relations (2.16), (2.17), the equation
for the $\mathcal{A}$-valued vector field $\mathbb{Y}$ reads
\begin{equation}
i[\mathbb{S},\mathbb{Y}]=JD\mathbb{S}.
\end{equation}
By taking (2.10) into account one derives the solution
\begin{equation}
\mathbb{Y}=\tilde{\mathbb{Y}}+\underline{\mathbb{Y}},\qquad
\tilde{\mathbb{Y}}\overset{\text{def}}{=}J(D\mathbb{S})^{\#},\qquad
[\mathbb{S},\underline{\mathbb{Y}}]=0,
\end{equation}
where by $\#$ we denote the operation
\begin{equation}
\mathbb{T}^{\#}\overset{\text{def}}{=}\frac1{2\pi}
\int^{2\pi}_0\exp\{it\mathbb{S}\}\cdot\mathbb{T}\cdot\exp\{-it\mathbb{S}\}(t-\pi)\,dt.
\end{equation}

To prove formula (3.2) one needs to mention the relation
\begin{equation}
(D\mathbb{S})^{\&}=0,
\end{equation}
where the sign $\&$ denotes the averaging operation
\begin{equation}
\mathbb{T}^{\&}\overset{\text{def}}{=}
\int^{2\pi}_0\exp\{it\mathbb{S}\}\cdot\mathbb{T}\cdot\exp\{-it\mathbb{S}\}\,dt.
\end{equation}
Identity (3.4) follows from (2.10) by applying the derivatives with respect to~$x$.

Let us note that operations (3.3) and (3.5) are related to each other as follows:
\begin{gather*}
(\mathbb{T}^{\#})^{\&}=(\mathbb{T}^{\&})^{\#}=0,\qquad
(\mathbb{T}^{\&})^{\&}=\mathbb{T}^{\&},
\\
i[\mathbb{S},\mathbb{T}]^{\#}=i[\mathbb{S},\mathbb{T}^{\#}]=\mathbb{T}-\mathbb{T}^{\&},
\qquad
[\mathbb{S},\mathbb{T}^{\&}]=[\mathbb{S},\mathbb{T}]^{\&}=0.
\end{gather*}
In particular, in (3.2) we have
\begin{equation}
\underline{\mathbb{Y}}=\mathbb{Y}^{\&}.
\tag{3.2a}
\end{equation}

Note that the $\mathcal{A}$-values vector field $\underline{\mathbb{Y}}$ in (3.2)
whose components commute with $\mathbb{S}$ cannot be computed from equation (3.1).
There must be some other conditions to determine this field.
We describe them in Theorem~3.1 below.

Now let us look at commutators between noncommutatively deformed slow coordinates.
From (2.16) one derives
\begin{equation}
[\hat{\mathbb{X}}_\varepsilon^j,\hat{\mathbb{X}}_\varepsilon^l]
=-i\varepsilon J^{jl}
-i\varepsilon^2\widehat{\mathbb{K}^{jl}}+O(\varepsilon^3),
\end{equation}
where
\begin{equation}
\mathbb{K}^{jk}=\tilde{\mathbb{K}}^{jk}+\underline{\mathbb{K}}^{jk},
\end{equation}
\begin{equation}
\tilde{\mathbb{K}}^{jl}\overset{\text{def}}{=}
J^{jm}\nabla_m\tilde{\mathbb{Y}}^{l} -J^{lm}\nabla_m\tilde{\mathbb{Y}}^{j}
-i[\tilde{\mathbb{Y}}^{j},\tilde{\mathbb{Y}}^{l}],
\end{equation}
and
\begin{equation}
\underline{\mathbb{K}}^{jl}\overset{\text{def}}{=}
J^{jm}\nabla_m\underline{\mathbb{Y}}^l-J^{lm}\nabla_m\underline{\mathbb{Y}}^j.
\end{equation}
Here the summation by repeated up and low indices is assumed,
and we use the covariant derivative notation
$$
\nabla\mathbb{T}\overset{\text{def}}{=}
D\mathbb{T}+iJ^{-1}[\tilde{\mathbb{Y}},\mathbb{T}].
$$

Since the fast spectrum does not degenerate,
it follows from (3.2) that $\underline{\mathbb{Y}}^j$
are functions of $\mathbb{S}$, i.e.,
\begin{equation}
\underline{\mathbb{Y}}^j(x)=y^j(\mathbb{S}(x),x).
\end{equation}
Therefore, from (3.9) and taking into account the properties
\begin{equation*}
(\nabla_l\mathbb{T})^{\&}=\nabla_l(\mathbb{T}^{\&}),\qquad
(\nabla_l\mathbb{T})^{\#}=\nabla_l(\mathbb{T}^{\#}),\qquad
\nabla\mathbb{S}=0,
\end{equation*}
we obtain
\begin{equation}
\underline{\mathbb{K}}^{jl}=J^{jm}D_m y^l-J^{lm}D_m y^j
\end{equation}
and, in particular,
\begin{equation}
[\underline{\mathbb{K}}^{jl},\mathbb{S}]=0.
\end{equation}

On the other hand, by applying the commutator with
$\hat{\mathbb{S}}_\varepsilon$ to both sides of (3.6) and taking (2.17) into account,
we see that the whole $\mathcal{A}$-valued tensor $\mathbb{K}^{jl}$ in (3.6)
also commutes with the action $[\mathbb{K}^{jl},\mathbb{S}]=0$.
Thus from (3.12) and (3.7) we derive $[\tilde{\mathbb{K}}^{jl},\mathbb{S}]=0$, and so,
\begin{equation}
\tilde{\mathbb{K}}^{jl}(x)=\tilde{k}^{jl}(\mathbb{S}(x),x).
\end{equation}

Note that by applying the averaging operation (3.5) to both sides of (3.8),
we can kill the first two summands in the right-hand side of (3.8)
and obtain the representation for the contravariant tensor
$\tilde{k}=\tilde{k}(\mathbb{S},x)$ (3.13):
\begin{equation}
\tilde{k}=-i[\tilde{\mathbb{Y}},\tilde{\mathbb{Y}}]^{\&}
= iJ[D\mathbb{S}^{\#},D\mathbb{S}^{\#}]^{\&} J.
\end{equation}

Now one can introduce the covariant tensor
\begin{equation}
\omega\overset{\text{def}}{=} J^{-1}\cdot \tilde{k}\cdot J^{-1}.
\end{equation}
Then by applying the commutator with $\hat{\mathbb{X}}^m_\varepsilon$
to both sides of  3.6) and using the Jacobi identity for double commutators we
conclude that the differential 2-form
(similar to the well-known Berry--Simon curvature form~[5, 6])
\begin{equation}
\Omega =\frac12\omega_{lj}dx^j\wedge dx^l
\end{equation}
is closed on $\mathcal{D}$, and thus it is exact:
\begin{equation}
\Omega=d\theta,\qquad
\omega_{lj}=D_j\theta_l-D_l\theta_j.
\end{equation}

Therefore by choosing in (3.10), (3.11)
\begin{equation}
y=J\theta
\end{equation}
and using representation (3.7) one obtains $\mathbb{K}=0$.

Thus relations (3.6) read
\begin{equation}
[\hat{\mathbb{X}}^j_\varepsilon,\hat{\mathbb{X}}^l_\varepsilon]
=-i\varepsilon J^{jl}+O(\varepsilon^3).
\end{equation}

We prove the following theorem.

\begin{theorem}
Let us determine the differential $2$-form {\rm(3.16)}
on the slow phase space by choosing $\omega_{jk}(s,x)$ from relations {\rm(3.14), (3.15)}.
Then take the primitive $1$-form $\theta$ {\rm(3.17)},
define the vector field $y=y(s,x)$ on the slow phase space by {\rm(3.18)},
and choose the operator-values vector field $\underline{\mathbb{Y}}$
in {\rm(3.2)} from {\rm(3.10)}.
Then the noncommutatively deformed coordinates {\rm(2.16)}
on the slow phase space satisfy, up to $O(\varepsilon^3)$,
the canonical commutation relations {\rm(3.19)} given
by the tensor $J$,
as well as commute with the quantum action
$\hat{\mathbb{S}}_\varepsilon$ up to $O(\varepsilon^2)$
{\rm(2.17)}.
\end{theorem}

Note that the transformation $x\to\mathbb{X}_\varepsilon$
can be interpreted as a nearly identical map of the original slow phase space
to a new one embedded into the whole fast--slow fibration.

\section{Computation of the quantum action}

Now we compute the correlation $\mathbb{W}$
to the action operator $\mathbb{S}$ in~(2.13).
The basic equation is (2.14).
From the general Lemma~A.1 (see in Appendix) one derives the formula for the Weyl symbol
of the exponent (2.14):
\begin{equation}
\exp\{2\pi i\hat{\mathbb{S}}_\varepsilon\}=\hat{\mathbb{U}}_\varepsilon,
\end{equation}
where
\begin{equation}
\mathbb{U}_\varepsilon=\exp\Big\{2\pi i\big(\mathbb{S}+\varepsilon\mathbb{W}
-i\varepsilon L+O(\varepsilon^2)\big)\Big\}1(x),
\end{equation}
$$
L\overset{\text{def}}{=} D\mathbb{S}(x)\cdot J D/2.
$$

After expansion of the exponential function
in (4.2) by known formulas of the perturbation theory, we obtain
\begin{equation*}
\mathbb{U}_\varepsilon=e^{2\pi i\mathbb{S}}
\Big(1-i\varepsilon L^{\&}
+\varepsilon\mathbb{W}^{\&}+O(\varepsilon^2)\Big)1,
\end{equation*}
where the operation $\&$ is determined by (3.5).
By taking identity (2.10) into account one computes
\begin{equation}
\mathbb{U}_\varepsilon=1+\varepsilon\mathbb{W}^{\&}
-\frac{i\varepsilon}{4\pi}
\int^{2\pi}_{0}e^{it\mathbb{S}}D\mathbb{S}JD(e^{-it\mathbb{S}})\,dt+O(\varepsilon^2).
\end{equation}

The basic equation (2.14) and relation (4.1) imply
$\mathbb{U}_\varepsilon\equiv 1+O(\varepsilon^2)$
and so it follows from (4.3) that
\begin{equation}
\mathbb{W}^{\&}=\frac{i}{4\pi}
\int^{2\pi}_{0}e^{it\mathbb{S}}D\mathbb{S}JD(e^{-it\mathbb{S}})\,dt
\end{equation}

From equation (3.1) we know that
\begin{equation}
JD(e^{-it\mathbb{S}})=-i[\mathbb{Y},e^{-it\mathbb{S}}].
\end{equation}
Thus the integral in (4.4) is transformed to
$$
\frac1{4\pi}\int^{2\pi}_{0}e^{it\mathbb{S}}D\mathbb{S}
[\mathbb{Y},e^{-it\mathbb{S}}]\,dt
=\frac12(D\mathbb{S}\cdot\mathbb{Y})^{\&}-\frac12(D\mathbb{S})^{\&}\cdot\mathbb{Y}
=\frac12(D\mathbb{S}\cdot\mathbb{Y})^{\&},
$$
where we used (3.4)  in the last equality.

Therefore, (4.4) reads
$\mathbb{W}^{\&}=(D\mathbb{S}\cdot\mathbb{Y})^{\&}/2$.
From this relation we can reconstruct $\mathbb{W}$ but not uniquely:
\begin{equation}
\mathbb{W}=D\mathbb{S}\cdot\mathbb{Y}/2+\mathbb{T},
\qquad \mathbb{T}^{\&}=0.
\end{equation}
Actually, such a ``gauge''  nonuniqueness as in (4.6)
is natural for equations like (2.14).

\begin{theorem}
The $\varepsilon$-correction to the action operator in {\rm(2.13)} is given by
\begin{equation}
\mathbb{W}
=D\mathbb{S}\cdot J\cdot(D\mathbb{S})^{\#}/2
+i[\mathbb{S},\mathbb{B}],
\end{equation}
where the operation $\#$ is defined by {\rm(3.3)}.
\end{theorem}

\section{Effective adiabatic hamiltonian}

Let us return to formula (2.18).
By using (3.2), (3.10) and (4.7) one can transform the correction
of order $\varepsilon$ as follows:
\begin{equation}
\hat{\mathbb{H}}=f(\hat{\mathbb{S}}_\varepsilon,\hat{\mathbb{X}}_\varepsilon)
-\varepsilon\hat{\mathbb{M}}+O(\varepsilon^2),
\end{equation}
where
\begin{equation}
\mathbb{M}\overset{\text{def}}{=}
(D\mathbb{H}(x)+D f(\mathbb{S},x))\cdot\mathbb{Y}/2
-((D\mathbb{S}\cdot \mathbb{Y})-(D\mathbb{S}\cdot\mathbb{Y})^{\&})
\cdot \partial f(\mathbb{S},x)/2
\end{equation}
and the operators in (5.2) are assumed to be Weyl-symmetrized.

The second summand on the right-hand side of (5.1)
still does not commute with the quantum action $\hat{\mathbb{S}}_\varepsilon$,
but one can apply to it the algebraic averaging procedure.

Let us define the following $\mathcal{A}$-values function on the slow phase space
\begin{equation}
\mathbb{R}(x)=r(\overset{1}{\mathbb{S}(x)},\overset{3}{\mathbb{S}(x)};x)
\overset{2}{\mathbb{M}}{}^{\#},
\end{equation}
where
\begin{equation}
r(s,s';x)= \delta f(s,s';x)^{-1},\qquad
\delta f(s,s';x)\overset{\text{def}}{=}
\int^1_0\partial f(s\mu+s'(1-\mu),x)\,d\mu.
\end{equation}
The function $\delta f$ does not vanish
since $f$ was assumed to be strongly monotone.
The numbers $1$, $2$, $3$ over the operators in (5.3)
are pointing their order from right to left.

\begin{lemma}
The transformation by the unitary operator
\begin{equation}
\hat{\mathbb{V}}_\varepsilon\overset{\text{\rm def}}{=}
\exp\{i\varepsilon\hat{\mathbb{R}}\}
\end{equation}
eliminates the noncommuting part of the second summand
$\hat{\mathbb{M}}$ on the right-hand side of {\rm(5.1)},
namely,
\begin{equation}
\hat{\mathbb{H}}\cdot\hat{\mathbb{V}}_\varepsilon=
\hat{\mathbb{V}}_\varepsilon\cdot
\big(f_\varepsilon(\hat{\mathbb{S}}_\varepsilon,\hat{\mathbb{X}}_\varepsilon)
+O(\varepsilon^2)\big).
\end{equation}
Here the function $f_\varepsilon(s,x)$ is given by
\begin{equation}
f_\varepsilon\overset{\text{\rm def}}{=}
f-\varepsilon(m+y\cdot Df),
\end{equation}
where
\begin{equation}
m(\mathbb{S}(x),x)= \big(D\mathbb{H}(x)\cdot\tilde{\mathbb{Y}}(x)\big)^{\&}/2
=\big(D\mathbb{H}(x)JD\mathbb{S}(x)^{\#}\big)^{\&}/2.
\end{equation}
\end{lemma}

\begin{proof}
Obviously, from (5.1) we have
\begin{equation}
\hat{\mathbb{V}}_\varepsilon^{-1}\cdot\hat{\mathbb{H}}\cdot\hat{\mathbb{V}}_\varepsilon
=f(\hat{\mathbb{S}}_\varepsilon
+i\varepsilon[\widehat{\mathbb{S},\mathbb{R}}],\hat{\mathbb{X}}_\varepsilon)
-\varepsilon\hat{\mathbb{M}}/2+O(\varepsilon^2),
\end{equation}
where all operators on the right-hand side are Weyl-symmetrized.
Now the first summand on the right can be expanded in $\varepsilon$:
\begin{equation*}
f(\hat{\mathbb{S}}_\varepsilon
+i\varepsilon[\widehat{\mathbb{S},\mathbb{R}}],\hat{\mathbb{X}}_\varepsilon)
=f(\hat{\mathbb{S}}_\varepsilon,\hat{\mathbb{X}}_\varepsilon)
+i\varepsilon\delta
f(\overset{1}{\hat{\mathbb{S}}},\overset{3}{\hat{\mathbb{S}}};\overset{4}{\hat{x}})
[\overset{2}{\widehat{\mathbb{S},\mathbb{R}}}]
+O(\varepsilon^2).
\end{equation*}
If we substitute this expansion into (5.9), then the summands of order $\varepsilon$
are combined into $\widehat{\mathbb{M}^{\&}}/2$
since $\mathbb{R}$ (5.3) satisfies the identity
$$
i\delta f(\overset{1}{\mathbb{S}(x)},\overset{3}{\mathbb{S}(x)};x)
[\overset{2}{\mathbb{S}(x),\mathbb{R}(x)}]
=(\mathbb{M}-\mathbb{M}^{\&})/2.
$$
From (5.2), (3.2), (3.4) we see that
\begin{align*}
\mathbb{M}^{\&}
&=(D\mathbb{H}\cdot\mathbb{Y})^{\&}/2 +Df\cdot\mathbb{Y}^{\&}/2
\\
&=(D\mathbb{H}\cdot\tilde{\mathbb{Y}})^{\&}/2
+D\mathbb{H}^{\&}\cdot \underline{\mathbb{Y}}/2 +Df\cdot \underline{\mathbb{Y}}/2
\\
&=(D\mathbb{H}\cdot\tilde{\mathbb{Y}})^{\&}/2+Df\cdot\underline{\mathbb{Y}},
\end{align*}
and therefore it follows from (3.10), (5.8) that
$$
\mathbb{M}^{\&}=m(\mathbb{S},x) + y(\mathbb{S},x)\cdot Df(\mathbb{S},x).
$$
Thus, we derive the expression for the $\varepsilon$-order summand in (5.7).
The lemma is proved.
\end{proof}

Formula (5.6) finalizes the quantum adiabatic approximation procedure up to $O(\varepsilon^2)$.
Since the quantum coordinates $\hat{\mathbb{X}}_\varepsilon$ commute with the quantum action
$\hat{\mathbb{S}}_\varepsilon$, we can sit onto the $k$th eigensubspace of this action.

\begin{theorem}
{\rm(a)} Let the quantum action operator
$\hat{\mathbb{S}}_\varepsilon=\mathbb{S}_\varepsilon(\hat{x})$
be determined by {\rm(2.13), (4.7)}.
Then its eigenvalues, up to $O(\varepsilon^2)$, are integer:
$k=O(\varepsilon^2)$, $k=0,1,2,\dots$.

{\rm{b}}
By the unitary transformation {\rm(5.5)}
and by the restriction to the $k$th eigensubspace of the quantum action
$\hat{\mathbb{S}}_\varepsilon$,
the original Hamiltonian $\hat{\mathbb{H}}$ is transformed to
the following adiabatic effective Hamiltonian
\begin{equation}
f_\varepsilon(k,\hat{\mathbb{X}}_\varepsilon)=\lambda_k(\hat{\mathbb{X}}_\varepsilon)
+\varepsilon\mu_k(\hat{\mathbb{X}}_\varepsilon)+O(\varepsilon^2).
\end{equation}
Here $\lambda_k$ {\rm(2.11)} is the leading adiabatic term, and
\begin{equation}
\mu_k(x)\overset{\text{\rm def}}{=} -m(k,x)-y(k,x) D\lambda_k,
\end{equation}
where the function $m$ is determined from {\rm(5.8)}
and the vector $y$ is determined from {\rm(3.18), (3.17), (3.15), (3.14)}.
\end{theorem}

The noncommutative-deformed slow coordinates $\hat{\mathbb{X}}^j_\varepsilon$
obey ($\operatorname{mod}O(\varepsilon^2)$) the Heisenberg commutation relations
(3.19) similar to~(2.4)
and so the effective Hamiltonian (5.9) is just equivalent to expression (2.8)
over the algebra (2.4).
In this sense, one can say that
$\hat{\mathbb{H}}$ is ``reduced'' to $\hat{\lambda}_k+\varepsilon\hat{\mu}_k+O(\varepsilon^2)$.
But
it is important to stress that the deformed coordinates
$\hat{\mathbb{X}}^j_\varepsilon$ in (5.10)
are essentially different from the original slow coordinates $\hat{x}^j$ (2.3).
The exchange $\hat{x}\to\hat{\mathbb{X}}_\varepsilon$ appeared in (5.10)
reflects an important effect, namely,
the noncommutative deformation of slow phase manifold accompanying
the separation of adiabatic bands in the system.

Note that the deformation of the slow phase space
and the deformation of the adiabatic terms $\lambda_k$ in (5.10) are
algebraically implemented by variations of the fast action $\mathbb{S}$
and do not depend on concrete realizations of the fast Hilbert space,
i.e., on the concrete choice of fast eigenvectors.
This makes our ``action'' approach consistent with
the Hamiltonian adiabatic approximation
schemes used in classical mechanics.

\section*{Appendix A}

\noindent{\bf Lemma A.1}
{\it
The following formulas for quantum composite function holds:
$a(\hat{b})=\hat{c}$.
Here
$c(x)=a(B)1(x)$,
the function $1(x)$ just identically equals the unity,
and the operator $B$ is determined as follows:
$$
B\overset{\text{\rm def}}{=}b\Big(x-\frac{i\varepsilon}{2}JD\Big)
=b(x)-\frac{i\varepsilon}{2}Db\cdot JD+O(\varepsilon^2).
$$}
\medskip

This general lemma implies
\begin{equation}
f(\hat{\mathbb{S}})=\hat{\mathbb{F}},\qquad
\mathbb{F}(x)=f\Big(\mathbb{S}(x-i\varepsilon JD/2)\Big)1(x).
\tag{A.1}
\end{equation}
To simplify the notation,
we here ignore the dependence of the functions $f$ (2.12)
on the coordinates $x$.
Let us expand the operator-values symbol $\mathbb{F}$ in (A.1)
into the power series in the small parameter $\varepsilon$:
\begin{equation}
\mathbb{F}=f\big(\mathbb{S}-i\varepsilon D\mathbb{S}\cdot J\cdot D/2
+O(\varepsilon^2)\big)1
=f(\mathbb{S})-i\varepsilon \overset{2}{L}\cdot
\delta f(\overset{3}{\mathbb{S}},\overset{1}{\mathbb{S}})1
+O(\varepsilon^2),
\tag{A.2}
\end{equation}
where
$L\overset{\text{def}}{=}D\mathbb{S}\cdot J\cdot D/2$
and the numbers 1,2,3 over the operators point their order.

Now we transfer the differential operators $D$ staying in the expressions for $L$
to the very right place:
\begin{equation}
\overset{2}{L}\cdot\delta f(\overset{3}{\mathbb{S}},\overset{1}{\mathbb{S}})
=J^{jl}\overset{2}{D_j\mathbb{S}}\cdot
\delta f(\overset{3}{\mathbb{S}},\overset{1}{\mathbb{S}})\cdot D_l/2
+J^{jl}\overset{4}{D_j\mathbb{S}}\cdot\overset{2}{D_l\mathbb{S}}\cdot
\delta^2f(\overset{5}{\mathbb{S}},\overset{3}{\mathbb{S}},\overset{1}{\mathbb{S}})/2.
\tag{A.3}
\end{equation}
Here by $\delta^2f$ we denote the second difference derivative
of the function $f$:
$$
\delta^2f(s,s',s'')\overset{\text{def}}{=}
\int^1_0d\tau \int^{\tau}_0d\mu\,\partial^2
f(s(1-\tau)+s'(\tau-\mu)+s''\mu).
$$

By applying (A.3) to the unity function $1(x)$ we obtain from (A.2):
$$
\mathbb{F}=f(\mathbb{S})-i\varepsilon
\overset{4}{D\mathbb{S}}\cdot J\cdot\overset{2}{D\mathbb{S}}\cdot
\delta^2f(\overset{5}{\mathbb{S}},\overset{3}{\mathbb{S}},\overset{1}{\mathbb{S}})/2
+O(\varepsilon^2)
$$
Then by using (3.1), we introduce the vector field $\mathbb{Y}$ into this expression:
\begin{align}
\mathbb{F}=f(\mathbb{S})-\varepsilon\overset{4}{D\mathbb{S}}\cdot
[\overset{2}{\mathbb{Y},\mathbb{S}}]\cdot
\delta^2f(\overset{5}{\mathbb{S}},\overset{3}{\mathbb{S}},\overset{1}{\mathbb{S}})/2
+O(\varepsilon^2)
\nonumber\\
=f(\mathbb{S})-\varepsilon
\big(\overset{2}{(D\mathbb{S}\cdot \mathbb{Y})}\cdot
\delta f(\overset{3}{\mathbb{S}},\overset{1}{\mathbb{S}})
-D(f(\mathbb{S}))\cdot\mathbb{Y}\big)/2.
\tag{A.4}
\end{align}
As the last step here we have used the general formulas
$$
[\mathbb{Y},g(\mathbb{S})]
=[\overset{2}{\mathbb{Y},\mathbb{S}}]\cdot
\delta g(\overset{3}{\mathbb{S}},\overset{1}{\mathbb{S}}),\qquad
D(f(\mathbb{S}))=\overset{2}{D\mathbb{S}}\cdot
\delta g(\overset{3}{\mathbb{S}},\overset{1}{\mathbb{S}}).
$$

From (A.4) and (A.1) one derives
$$
\widehat{f(\mathbb{S})}=f(\hat{\mathbb{S}})
+\varepsilon
\big((\widehat{D\mathbb{S}\cdot \mathbb{Y}})\partial f(\hat{\mathbb{S}})
-\widehat{D(f(\mathbb{S}))}\cdot\hat{\mathbb{Y}}\big)/2
+O(\varepsilon^2),
$$
where on the right-hand side the operators $D\mathbb{S}\cdot \mathbb{Y}$
and $\mathbb{S}$ are Weyl-symmetrized.

If we recall that the function $f$ actually depends on the coordinates $x$,
then this changes the reminder $O(\varepsilon^2)$ only
in the later formula:
\begin{equation}
\hat{\mathbb{H}}=\widehat{f(\mathbb{S},x)}
=f(\hat{\mathbb{S}},\hat{x})+\varepsilon\hat{\mathbb{Q}}+O(\varepsilon^2),
\tag{A.5}
\end{equation}
where
\begin{equation}
\mathbb{Q}=(D\mathbb{S}\cdot\mathbb{Y})\cdot\partial f(\mathbb{S},x)/2
-D\mathbb{H}\cdot \mathbb{Y}/2+Df(\mathbb{S},x)\mathbb{Y}/2.
\tag{A.6}
\end{equation}
All the operators in (A.5) and (A.6) are assumed to be Weyl-symmetrized.

The obtained formula (A.5) was used in (2.15).

\section*{Acknowledgments}

This study (research grant N 12-01-0140) was supported by The National
Research University - Higher School of Economics; Academic Fund Programm.

\end{document}